\documentclass{article}

\usepackage{arxiv}
\usepackage{float}
\usepackage[utf8]{inputenc}
\usepackage{hyperref}
\hypersetup{
	colorlinks=true, 
	citecolor=blue,
	linkcolor=black,
	urlcolor=blue
}
\usepackage{setspace}
\usepackage{amsmath,amssymb,amsthm}
\usepackage{graphicx}
\usepackage{enumerate}

\usepackage{enumitem}
\usepackage{algorithm}
\usepackage[toc,page]{appendix}
\usepackage{algpseudocode}
\usepackage{tikz}
\usetikzlibrary{arrows, calc, backgrounds, shapes.misc, positioning, shapes, fit}

\newcommand{\figref}[1]{Figure~\ref{#1}}

\newcommand{\bmat}[1]{\begin{bmatrix}#1\end{bmatrix}}
\newcommand{\bmtx}{\begin{bmatrix}}
	\newcommand{\emtx}{\end{bmatrix}}
\newcommand{\bsmtx}{\left[ \begin{smallmatrix}} 
	\newcommand{\esmtx}{\end{smallmatrix} \right]}
\newcommand{\field}[1]{\mathbb{#1}}
\newcommand{\R}{\field{R}}

\newcommand{\Sm}{\field{S}}

\makeatletter

\newcommand{\Rmnum}[1]{\expandafter\@slowromancap\romannumeral #1@}
\makeatother
\definecolor{mycolor1}{HTML}{367D7D}
\definecolor{mycolor2}{HTML}{D33502}
\definecolor{mycolor3}{HTML}{FAA818}
\definecolor{mycolor4}{HTML}{41A30D}
\definecolor{mycolor5}{HTML}{FFCE38}
\definecolor{mycolor6}{HTML}{6EBCBC}
\definecolor{mycolor7}{HTML}{37526D}

\colorlet{plant}{mycolor5!40!white}
\colorlet{unc}{mycolor6!40!white}
\colorlet{closedloop}{gray!20!white}
\colorlet{controller}{mycolor4!30!white}
\colorlet{dist}{mycolor3!60!white}
\colorlet{marker}{mycolor7}
\colorlet{algo}{mycolor3!60!white}
\colorlet{miscplot}{mycolor2}
\colorlet{clplot}{mycolor2}
\colorlet{filter}{magenta!20!white}

\newtheorem{theorem}{Theorem}

\newtheorem{ex}{Example}
\newtheorem{lem}{Lemma}
\newtheorem{cor}{Corollary}
\newtheorem{eremark}{Remark}

\tikzset{>=stealth'}
\title{An Efficient Algorithm to Compute Norms for Finite Horizon, Linear Time-Varying Systems}

\author{
  Jyot Buch\\
 Department of AEM\\
 University of Minnesota\\
 Minneapolis, MN 55455 \\
 \texttt{buch0271@umn.edu}\\
   \And
 Murat Arcak \\
  Department of EECS\\
  University of California\\
  Berkeley, CA 94720 \\
  \texttt{arcak@berkeley.edu} \\
  \And
 Peter Seiler\\
Department of EECS\\
University of Michigan\\
Ann Arbor, MI 48109 \\
\texttt{pseiler@umich.edu} \\
}

\begin{document}
\maketitle
\begin{abstract}
		We present an efficient algorithm to compute the induced norms of finite-horizon Linear Time-Varying (LTV) systems.  The formulation includes both induced $\mathcal{L}_2$ and terminal Euclidean norm penalties. Existing computational approaches include the power iteration and bisection of a Riccati Differential Equation (RDE). The power	iteration has low computation time per iteration but overall convergence can be slow. In contrast, the RDE condition provides guaranteed bounds on the induced gain but single RDE integration can be slow. The complementary features of these two algorithms are	combined to develop a new algorithm that is both fast and provides provable upper and lower bounds on the induced norm within the desired tolerance. The algorithm also provides a worst-case disturbance input that achieves the lower bound on the norm. We also present a new proof which shows that the power iteration for this problem converges monotonically. Finally, we show a controllability Gramian based simpler computational method for induced $\mathcal{L}_2$-to-Euclidean norm. This can be used to compute the reachable set at any time on the horizon. Numerical examples are provided to demonstrate the proposed algorithm.
\end{abstract}

% keywords can be removed
\keywords{Time-Varying Systems, Linear Systems, Numerical Algorithms.}

\section{Introduction}
This paper presents an algorithm to compute induced gains for  finite
horizon linear time-varying (LTV) systems. The goal is to improve
computational speed for such analysis.  This has a practical impact for
engineered systems that follow
a finite-horizon trajectory, including space launch vehicles and
robotic manipulators, that are often modeled with nonlinear
ordinary differential equations.  An approximate LTV model is obtained
by Jacobian linearization along the trajectory and induced-gains for the
LTV model can be used for analysis and design of the feedback systems.
Moreover, worst-case disturbances from the LTV analysis can be further
studied in the nonlinear model.

Two existing approaches for computing LTV induced norms are summarized in Section~\ref{sec:existingalgo}, as the power
iteration and the bisection method. The power iteration repeatedly
integrates the dynamics of the LTV system and a related adjoint
system. These steps are coupled by an alignment condition. 
This amplifies the input in the largest gain direction so that the iterates
converges to the worst-case disturbance input. The bisection method
relies on a related matrix Riccati Differential  Equation (RDE)
condition. Upper or lower bounds on the induced gain are obtained
based on the existence of a solution for the RDE. The induced gain can
be computed within a desired tolerance via bisection. These methods
have complementary properties as discussed in Section~\ref{sec:compissue}.  Each step
of the power iteration is fast but overall convergence can be  slow.
Moreover, the power iteration only computes a lower bound with no
guarantee on the gap to the induced gain. In comparison, the RDE
bisection method
provides guaranteed upper and lower bounds. However, a single RDE
integration is slow for higher order systems. 

There are three main contributions of the paper. First, we propose a
combined algorithm that utilizes the complementary
benefits of the power iteration and RDE bisection methods (Section~\ref{sec:combAlgo}). Second, we show that the power iteration algorithm converges
monotonically to the induced system norm (Section~\ref{sec:poweritconv}).  This is a
stronger convergence result than existing results for power
iterations on Hilbert spaces~\cite{eastman2007power}. Finally, we show that a simpler
computational condition can be obtained for the special case of
induced $\mathcal{L}_2$-to-Euclidean gain (Section~\ref{sec:l2toE}) using existing time-varying Gramian results (Section $22$ of~\cite{brockett2015finite}). 

Among the most closely related work is~\cite{tierno1997numerically}, which uses the power iteration to compute induced gains of uncertain
nonlinear, time-varying systems. The restriction to LTV systems, as done in our paper,  allows us to provide a convergence proof for the power iteration  and to use the
RDE to compute provable upper bounds on the induced gain. The work
in~\cite{imae1996h} also considers computing induced norms for finite-horizon LTV
systems. The algorithm in~\cite{imae1996h} relies on an iteration using a different
RDE. Again, no convergence proof is provided. Moreover, the
computational cost of the RDE in~\cite{imae1996h} scales similarly to the RDE in
the bisection method. Other related work includes computing the
induced $\mathcal{L}_2$ gain for LTI systems on infinite horizons~\cite{boyd1989bisection,bruinsma1990fast} and
finite horizons~\cite{bamieh2003computing}. 

\vspace{0.1in}
\noindent \textbf{Notation:} Let $\R^{n \times m}$ and $\Sm^{n}$ denote the sets of $n$-by-$m$ real matrices and $n$-by-$n$ real, symmetric matrices. 
Let $\mathcal{L}_2^n[0,T]$ denote the Hilbert space of Lebesgue integrable signals $v:[0,T] \rightarrow \R^n$ with inner product $\langle v, v \rangle := \int_0^T v(t)^\top v(t) \, dt$.  The inner product defines a norm $\|v\|_{2, [0,T]} = \sqrt{\langle v, v \rangle}$.  If $\|v\|_{2, [0,T]} < \infty$ then $v \in \mathcal{L}_2^n[0,T]$.

\section{Problem Formulation}
\label{sec:prob}

Consider an LTV System $G$ defined on horizon $[0,T]$:
\begin{align}
\label{eq:LTV1} \dot{x}(t) & = A(t)\, x(t) + B(t)\, d(t) \\	
\label{eq:LTV2} e(t) & = C(t)\, x(t) + D(t)\, d(t) 
\end{align}
where $x(t) \in \R^{n_x}$ is the state, $d(t) \in \R^{n_d}$ is the external input, and
$e(t) \in \R^{n_e}$ is the error output at time $t$. The state matrices 
$A:[0,T] \rightarrow \R^{n_x \times n_x}$,
$B:[0,T] \rightarrow \R^{n_x \times n_d}$,
$C:[0,T] \rightarrow \R^{n_e \times n_x}$, and
$D:[0,T] \rightarrow \R^{n_e \times n_d}$ are piecewise-continuous (bounded) matrix valued functions of time. Explicit time dependence of the 
state matrices will be omitted when it is clear from the context. It is assumed throughout that the horizon is finite i.e. $T < \infty$. The input, output and state dimensions of $G$ are assumed to be constant throughout the horizon. The performance of $G$ is assessed in terms of an induced gain
with two components of the output:
\begin{align}
\label{eq:output}
\bmtx e_I(t) \\ e_E(t) \emtx =
\bmtx C_I(t) \\ C_E(t) \emtx \, x(t) 
+ \bmtx D_I(t) \\ 0 \emtx \, d(t) 
\end{align}
where $e_I(t) \in \R^{n_I}$ and $e_E(t) \in \R^{n_E}$ with $n_e = n_E+n_I$.
The generalized performance measure of $G:\mathcal{L}^{n_d}_{2}[0,T] \rightarrow \R^{n_E} \oplus \mathcal{L}^{n_I}_{2}[0,T]$ is then defined as:
\begin{align}
\label{eq:def}
\|G\|_{[0,T]}:=
\sup_{\stackrel{0\ne d \in \mathcal{L}_2[0,T]}{x(0)=0}}
\left[
\frac{\|e_E(T)\|_2^2 + \|e_I\|^2_{2,[0,T]}} {\|d\|^2_{2,[0,T]}} 
\right]^{\frac{1}{2}}
\end{align}
This defines an induced gain from the input $d$ to a mixture of an
$\mathcal{L}_2$ and terminal Euclidean norm on the output $e$. This norm was previously used in \cite{khargonekar1991h_,buch2020FHRobSyn} to develop extensions of robust control for finite-horizon, time-varying systems. Note
that if $n_E=0$ then there is no terminal Euclidean norm penalty on
the output. This case corresponds to the standard, finite-horizon
induced $\mathcal{L}_2$ gain of $G$. Similarly, if $n_I=0$ then there
is no $\mathcal{L}_2$ penalty on the output. This case corresponds to
a finite-horizon $\mathcal{L}_2$-to-Euclidean gain. This can be used
to bound the terminal output $e_E(T)$ reachable by an
$\mathcal{L}_2$ disturbance input. Zero feed-through from $d$ to $e_E$
ensures that Euclidean penalty is well-defined at any time $t\in[0,T]$.
The system is assumed to be initially at rest, i.e. $x(0)=0$. Non-zero initial conditions similar to~\cite{khargonekar1991h_} can also be handled with minor modifications. We are interested in efficiently computing the bound on the induced gain $\|G\|_{[0,T]}$ within a specified numerical tolerance.

\section{Existing Methods For Gain Computation}
\label{sec:existingalgo}

\subsection{Optimal Control Formulation}
The induced norm defined in Eq.~\eqref{eq:def} is related to the following optimal control problem~\cite{seiler2019finite}:
\begin{align*}
J^{*} = &\sup_d \, \frac{1}{2}\left[ e_E(T)^\top e_E(T)  + \int_0^T e_I(t)^\top e_I(t) \,dt \right] \\ 
&s.t.\nonumber \mbox{ Eq. } \eqref{eq:LTV1}, \eqref{eq:output} \mbox{ with } x(0) = 0, \|d\|_{2,[0,T]} = 1 \nonumber
\end{align*}
The optimal cost and induced norm are related as $J^*=\frac{1}{2}\|G\|_{[0,T]}^2$. The standard Euler-Lagrange optimization framework~(Section $2.5$ of~\cite{bryson1975applied}) can be used to solve this problem. The Hamiltonian is defined using the quadratic performance cost and Lagrange multipliers for constraints. A two-point boundary-value problem is obtained from the necessary conditions for optimality. This leads to two numerical algorithms for computing the induced norm of $G$: (i) the power iteration, and (ii) bisection on a related Riccati Differential Equation.  These two algorithms are briefly summarized in the following subsections.

\subsection{Power Iteration}
\label{sec:powerit}

The power iteration is an iterative method to approximate the maximum eigenvalue and singular value for matrices within numerical tolerance~\cite{golub2013matrix,demmel1997applied}. It was studied in~\cite{tierno1997numerically} for computing the gain of uncertain, nonlinear time-varying systems. A version of the power iteration can also be used to approximate the induced gain of an LTV system within numerical tolerance. The adjoint (costate) dynamics $G^\sim:\R^{n_E} \oplus \mathcal{L}^{n_I}_{2}[0,T] \rightarrow \mathcal{L}^{n_d}_{2}[0,T]$ obtained from the two-point boundary-value problem are as follows:
\begin{align}
\label{eq:adj}\dot{p}(t) &= - A(t)^\top p(t) - C_I(t)^\top q(t)\\
\label{eq:adjop}r(t) &= B(t)^\top p(t) + D_I(t)^\top q(t)
\end{align}
with boundary condition: 
\begin{align}
\label{eq:adjbc}
p(T) = C_E(T)^\top e_E(T)
\end{align}
where $p(t) \in \R^{n_x}, q(t) \in \R^{n_I}, r(t) \in \R^{n_d}$ are adjoint state, input and output respectively. The specific steps for power iteration are outlined in Algorithm~\ref{alg:power}. 
\begin{algorithm}[h]
	\linespread{1}\selectfont
	\caption{Power Iteration for LTV System} \label{alg:power}  
	\begin{algorithmic}[1]
		\State \textbf{Given:} $G$, $G^\sim$
		\State \textbf{Initialize:} $d^{(1)}$ with $\|d^{(1)}\|_{2,[0,T]} = 1$, $N$, $\epsilon_a$,  $\gamma^{(0)}=-\infty$ 
		
		\For{$i=1:N$} 
		
		\State \parbox[t]{0.9\linewidth}{\textbf{Forward Sim:} Simulate $G$ from $t=0$ to $t=T$ with $x(0)=0$ and input $d^{(i)}$ to generate $e_I^{(i)}$, $e_E^{(i)}$, and to compute forward performance $\gamma_{f}^{(i)}$ using Eq.~\eqref{eq:fwdperf}.}	
		
		\State \parbox[t]{0.9\linewidth}{\textbf{Backward Sim:} Simulate $G^\sim$ from $t=T$ to $t=0$ with $p(T) \equiv C_E(T)^\top e_E^{(i)}(T)$ and input $q^{(i)} \equiv e_I^{(i)}$ to generate $r^{(i)}$ using Eq.~\eqref{eq:adjop}.}		
		
		\State \parbox[t]{0.9\linewidth}{\textbf{Compute Gain:} $\gamma^{(i)} = \| r^{(i)} \|_{2,[0,T]}$.}		
		
		\State \textbf{Alignment Condition:} $d^{(i+1)} =  r^{(i)}/\gamma^{(i)}$.		
		
		\State \parbox[t]{0.9\linewidth}{\textbf{Stop Condition:} Terminate if $\gamma^{(i)} - \gamma^{(i-1)} < \epsilon_a$.}
		
		\EndFor
		
		\State \parbox[t]{0.9\linewidth}{\textbf{Compute Output:} Simulate $G$ using  $d_{\pi} = d^{(i+1)}$ and compute forward performance $\gamma_{\pi}$ using Eq.~\eqref{eq:fwdperf}.}
		\vspace{0.02in}	
	
		\State \textbf{Output:} $\gamma_{\pi}$, $d_{\pi}$.
	\end{algorithmic}
\end{algorithm}

Let system $G$ and its adjoint $G^\sim$ be given. The algorithm is initialized with a candidate disturbance $d^{(1)}$, maximum number of iterations $N$, initial performance $\gamma^{(0)}=-\infty$ and desired absolute tolerance $\epsilon_a$. The initial disturbance can be chosen randomly or by any other means. It is normalized so that $\|d^{(1)}\|_{2,[0,T]} = 1$. The first step in the iteration is to simulate the system $G$ forward in time from $t=0$ to $t=T$ using chosen disturbance $d^{(i)}$ and zero initial conditions. This step yields the outputs $e_I^{(i)}$ and $e_E^{(i)}$. The forward induced performance $\gamma_f^{(i)}$ can be computed for this specific unit-norm disturbance.
\begin{align}
\label{eq:fwdperf}
\gamma_{f}^{(i)} := \left[\|e_E^{(i)}(T)\|_2^2 + \|e_I^{(i)}\|^2_{2,[0,T]}\right]^{\frac{1}{2}}
\end{align}
The next step involves backward simulation of the adjoint dynamics $G^\sim$ from $t=T$ to $t=0$ using input $q^{(i)}\equiv e_I^{(i)}$ and boundary condition for $p(T) \equiv C_E(T)^\top e_E^{(i)}(T)$. This gives the adjoint output $r^{(i)}$. The gain $\gamma^{(i)}$ is given by the norm of the signal $r^{(i)}$. The disturbance for the next iteration $d^{(i+1)}$ is obtained by normalizing $r^{(i)}$. The iterations are terminated if the performance $\gamma$ fails to increase by more than the specified absolute tolerance. In addition, the iterations also stop if the number of iterations exceed specified $N$. The final outputs are induced gain $\gamma_\pi$ and the corresponding worst-case disturbance $d_\pi$. The repeated evaluations of $G$ and it's adjoint $G^\sim$ amplify the disturbance.
Moreover, the alignment condition ensures that the disturbance $d^{(i)}$ does not blow up and gets aligned along the direction of the largest gain as iteration progresses. It is known that the related matrix power iteration for singular values converges to the induced Euclidean norm of the matrix under mild technical conditions~\cite{golub2013matrix,demmel1997applied}. In Section~\ref{sec:poweritconv}, we show that the performance $\gamma^{(i)}$ is monotonically non-decreasing and power iteration for LTV system $G$ converges to the induced gain. If the iterative loop in Algorithm~\ref{alg:power} is terminated early, then $\gamma_{\pi}$ is a lower bound on the actual gain $\|G\|_{[0,T]}$. Numerical integration issues may arise on a significantly long horizon if $G$ is an unstable LTI system.

\subsection{Bisection on the Riccati Differential Equation (RDE)}
\label{sec:RDE}

Due to linear dynamics, the adjoint solution can be obtained as $p(t) = P(t) \, x(t)$ where $P(t)$ is a time-varying solution to a related RDE. The next theorem states an equivalence between a bound on the performance $\|G\|_{[0,T]}$ and the existence of a solution to a related RDE \cite{khargonekar1991h_,green2012linear,moore15,seiler2019finite}. 
\begin{theorem}
	\label{thm:nominalperf}
	Consider an LTV system \eqref{eq:LTV1} with $\gamma>0$ given. Let
	$Q : [0,T] \rightarrow \Sm^{n_x}$,
	$S : [0,T] \rightarrow \R^{n_x \times n_d}$,
	$R : [0,T] \rightarrow \Sm^{n_d}$, and $F \in \R^{n_x\times n_x}$ be
	defined as follows.
	\begin{align*}
	\begin{split}
	Q := C_I^\top C_I, \hspace{0.1in}  
	S := C_I^\top D_I, \hspace{0.1in}
	R := D_I^\top D_I-\gamma^2 I_{n_d},  \hspace{0.1in} 
	F := C_{E}(T)^\top C_{E}(T)
	\end{split}
	\end{align*}
	The following statements are equivalent:
	\begin{enumerate}
		\item $\|G\|_{[0,T]} < \gamma$
		\item $R(t) < 0$ for all $t \in [0,T]$. Moreover, there
		exists a differentiable function $P:[0,T] \rightarrow \Sm^{n_x}$ such that
		\begin{align*}
		\dot{P} +  A^\top P + PA +Q - (PB+S) R^{-1} (PB+S)^\top = 0, \hspace{0.2in} P(T)=F
		\end{align*}
		This is a Riccati Differential Equation (RDE).
	\end{enumerate}  
\end{theorem}
The nominal performance $\|G\|_{[0,T]} < \gamma$ is achieved if the
associated RDE solution exists on $[0,T]$ when integrated backward
from $P(T)=F$. The assumption $R(t)<0$ ensures $R(t)$ is invertible and hence the RDE is well-defined $\forall t\in [0,T]$.
Thus, the solution of the RDE exists on $[0,T]$ unless it grows
unbounded. The smallest bound on $\gamma$ is computed using bisection as summarized in Algorithm~\ref{alg:bisectAlgo}.
\begin{algorithm}[h]
	\linespread{1}\selectfont
	\caption{RDE Bisection Method} \label{alg:bisectAlgo}  
	\begin{algorithmic}[1]
		\State \textbf{Given:} $G$ 
		\State \textbf{Initialize:}  $\epsilon_a$, $\gamma_{lb}$, $\gamma_{ub}$ with $\gamma_{lb} \le \|G\|_{[0,T]} \le \gamma_{ub}$
		
		\While{$\gamma_{ub} - \gamma_{lb} > \epsilon_a$}
		
		\State \textbf{Bisect:}  $\gamma_{try} = 0.5\,(\gamma_{ub} + \gamma_{lb})$
		\vspace{0.02in}
		
		\State \parbox[t]{0.9\linewidth}{\textbf{Integrate RDE:} Solve the RDE with $\gamma_{try}$ backwards in time from $P(T)=F$}
		\vspace{0.02in}
		
		\State \textbf{Update:} If $P(0)<\infty$ then $\gamma_{ub} = \gamma_{try}$ else $\gamma_{lb} = \gamma_{try}$				
		
		\EndWhile
		
		\State \textbf{Output:} $\gamma_{ub}$, $\gamma_{lb}$.
	\end{algorithmic}
\end{algorithm}

This algorithm is initialized with bounds on the gain such that $\gamma_{lb} \le \|G\|_{[0,T]} \le \gamma_{ub}$.  The simplest choice for the lower bound is $\gamma_{lb} = \max_{t \in [0,T]} \bar \sigma (D_I(t))$. This can be computed (approximately) on a dense time grid. An upper bound can be found (if one is not known) by choosing increasing values of $\gamma$ until the RDE solution exists on $[0,T]$. The gain $\|G\|_{[0,T]}$ is finite and hence a finite upper bound will exist. Every bisection step involves integrating the RDE with $\gamma_{try}$ backward from $P(T)=F$. The bisection continues until the bounds are within the specified  tolerance $\epsilon_a$. The RDE solution grows unbounded for each lower bound step, i.e. it only exists on $(t^*,T]$ for some $t^*>0$.  The incomplete RDE solution for any $\gamma_{lb}$ can be used to construct a specific disturbance $d_{lb}$ such that it achieves the induced gain $\gamma_{lb}$~\cite{iannelli2019worst}. This disturbance $d_{lb}$ provides a verification that the induced gain $\|G\|_{[0,T]}$ is at least $\gamma_{lb}$. The RDE exhibits numerical integration issues on a significantly long horizon if $G$ is an unstable LTI system.

\section{Proposed Algorithm}
\label{sec:proposedalgo}

\subsection{Convergence of Power Iteration}
\label{sec:poweritconv}

The convergence properties of power iterations for finite-dimensional matrices are summarized first. The eigenvalue power iteration for a square matrix $M \in \R^{n \times n}$ involves iterations of the form: $v^{(i+1)} = M v^{(i)} / \| M v^{(i)}\|$.   Let $\lambda_1, \ldots, \lambda_n$ be the eigenvalues of $M$ ordered from largest magnitude to smallest.  If $|\lambda_1|>|\lambda_2|$ and $v^{(1)}$ is chosen randomly then $v^{(i)}$ converges (with probability $1$) to the eigenvector associated with $\lambda_1$ at a convergence rate of $\frac{|\lambda_2|}{|\lambda_1|}$ (Theorem $5.6$ of~\cite{quarteroni2010numerical}). However, the iteration may fail to converge if $|\lambda_1|=|\lambda_2|$ but $\lambda_1\neq \lambda_2$ (Theorem $2$ of~\cite{parlett1973geometric}, Section $7.3.1$ of~\cite{golub2013matrix} and related references). 

Similarly, the singular value power iteration for a matrix $M \in \R^{n \times m}$ involves iterations of the form: $v^{(i+1)} = M^\top M v^{(i)} / \| M^\top M v^{(i)}\|$. This can be viewed as an eigenvalue power iteration on the matrix $M^\top M$. However a slightly stronger convergence result is obtained because $M^\top M$ is Hermitian and positive semidefinite. Specifically,  If $v^{(1)}$ is chosen randomly then $v^{(i)}$ converges (with probability $1$) to the space spanned by the right singular vectors associated with the largest singular values~(Theorem $8.2.1$ of~\cite{golub2013matrix}). The convergence of the finite-dimensional singular value power iteration does not require the largest singular value to have multiplicity $1$ (Theorem $3.8$, Chapter $4$ of~\cite{stewart2001matrix}).

Next, consider the power iteration for the LTV systems in Algorithm~\ref{alg:power}. This involves a forward simulation with $G$ followed by a backward simulation with the adjoint $G^\sim$.  This combined action can be written as the composition $G^\sim G:\mathcal{L}^{n_d}_{2}[0,T]\rightarrow \mathcal{L}^{n_d}_{2}[0,T]$. Algorithm~\ref{alg:power} is thus analogous to the singular value power iteration but on an infinite dimensional space. To make this precise, first note that the operator $G^\sim G$ is a non-negative, self-adjoint, and bounded linear operator on the Hilbert space $\mathcal{L}^{n_d}_{2}[0,T]$. Theorem $3.1$ in~\cite{eastman2007power} provides a convergence proof for power iterations on a bounded linear operator in Hilbert space. The corollary stated below follows from this result.
\begin{cor}
	Suppose $G^\sim G$ has a single, simple dominant eigenvalue $\lambda_1$ and corresponding eigenvector $\psi_1$ with $\|\psi_1\|_{2,[0,T]}=1$. If $\langle d^{(1)}, \psi_1 \rangle \ne 0$ then $\| d^{(i)} - \psi_1\|_{2,[0,T]} \rightarrow 0$ as $i\rightarrow \infty$.
\end{cor}

Next, we present a stronger convergence result for Algorithm~\ref{alg:power} which exploits the structure in singular value power iteration on $\mathcal{L}^{n_d}_{2}[0,T]$.  We do not require the assumption of a single, simple dominant eigenvalue. However, we require that $G$ has no feedthrough.  This ensures that $G^\sim G$ is a compact operator (Lemma~\ref{lem:compactness} in Appendix~\ref{sec:compactness}). The eigenvalues of $G^\sim G$ are real and eigenvectors associated with distinct eigenvalues are orthogonal (Theorem 8.11 of~\cite{young1988introduction}).  In fact, $G^\sim G$  has an orthonormal set of eigenvectors $\{ \psi_k \}_{k=1}^\infty \in \mathcal{L}^{n_d}_2[0,T]$ and eigenvalues $\{ \lambda_k\}_{k=1}^\infty \in \R^+ \cup \{0\}$. Moreover, if $d \in \mathcal{L}^{n_d}_2[0,T]$ then the operation $r = G^\sim G(d)$ can be written as:
\begin{align}
\label{eq:infsum}
r = G^\sim G(d) = \sum_{k = 1}^{\infty} \lambda_k \langle d, \psi_k \rangle \psi_k 
\end{align}	
Assume that eigenvalues are sorted in descending order and use non-negativity of $G^\sim G$ to get $\lambda_k \geq \lambda_{k+1} \geq 0$.

\begin{theorem}
	\label{thm:powerconv}
	Assume $D_I=0$ so that $G^\sim G$ is compact. Further assume that the dominant eigenvalue has multiplicity $m$ and $d^{(1)}$ satisfies $\langle d^{(1)}, \psi_k \rangle \ne 0$ for some $k\in \{1,\ldots,m\}$. Then $\gamma^{(i+1)} \ge \gamma^{(i)}$, $\forall i \geq 1$ and as $i \rightarrow \infty$ we have:
	\begin{enumerate}[label=(\alph*)]
		\item $d^{(i)} \rightarrow$ span of $\{\psi_1,\ldots,\psi_m\}$
		\item $\gamma^{(i)} \rightarrow \|G\|_{[0,T]}^2$
	\end{enumerate}	
\end{theorem}
\begin{proof} 
	The proof is given in Appendix~\ref{sec:proof2}.
\end{proof}
Theorem~\ref{thm:powerconv} added the assumption $D_I=0$ to ensure compactness of $G^\sim G$.  This yields a stronger convergence result than obtained in~\cite{eastman2007power} for the eigenvalue power iteration. Specifically, the power iteration in Algorithm~\ref{alg:power} converges even if the dominant eigenvalue is repeated.  If $D_I\neq0$ then the power iteration will still converge by result in~\cite{eastman2007power} if $\lambda_1>\lambda_2$. However, this condition cannot be verified in practice as the eigenvalues are not known. Finally, Theorem~\ref{thm:powerconv} also guarantees that the iteration will be non-decreasing. This is a useful diagnostic condition as $\gamma^{(i+1)} < \gamma^{(i)}$ indicates errors in the numerical integration. Algorithm~\ref{alg:power} can detect this anomaly, terminate the iteration, and warn the user to reduce the integration step size.

\subsection{Computational Issues}
\label{sec:compissue}

The computational time for bisection method can be approximated by the number of bisections $N_{RDE}$ multiplied by the computational time for one RDE integration $T_{RDE}$.\footnote{The RDE solution may grow unbounded for lower bound step and exist only on $(t^*,T]$ for some $t^*>0$. It takes longer to perform the integration over the entire horizon ($t^*=0$) than if the solution grows unbounded for some $t^*>0$. A more precise estimate of the computation time would account for the dependence of $T_{RDE}$ on the choice of $\gamma_{try}$ at each bisection step.} If the algorithm starts with a gap $\gamma_g := \gamma_{ub}-\gamma_{lb}$ then it takes at least $N_{RDE} = \log_2(\gamma_g/\epsilon_a)$ to achieve an absolute tolerance of $\epsilon_a$. Additional RDE integrations are required if an upper bound is not known at the start of the algorithm. Parallel computing resources can be exploited to integrate RDE on grid of $\gamma$. However, the single RDE integration time depends on many factors including the time horizon $T$, system order $n_x$, and integration solver tolerances. The cost of the RDE grows roughly linearly with $T$. The RDE is an $n_x\times n_x$ matrix differential equation and hence this integration requires solving $n_x (n_x+1)/2 \approx \mathcal{O}(n_x^2)$ scalar, nonlinear differential equations.\footnote{It is possible to solve the RDE by instead integrating a set of $2n_x \times n_x$ linear differential equations with a related Hamiltonian. This form does not seem to offer computational advantages.} Algorithm 2 guarantees on exit that $\|G\|_{[0,T]}$ has been computed within an absolute tolerance $\epsilon_a$.

The power iteration method has complementary properties. The computational time can be approximated as the number of iterations $N_{PI}$ multiplied by the time required for one iteration $T_{PI}$. A single power iteration requires the integration of $G$ followed by the integration of $G^\sim$. Both $G$ and $G^\sim$ are LTV systems of order $n_x$ and hence, $T_{PI}$ scales as $\mathcal{O}(n_x)$. It also scales linearly with the horizon $T$. The primary advantage of the power iteration is that $T_{PI}$ is typically significantly less than $T_{RDE}$. We expect the ratio $T_{RDE}/T_{PI}$ to grow with the state dimension as $\mathcal{O}(n_x)$.  Moreover, the power iteration is guaranteed to converge by Theorem~\ref{thm:powerconv}. However, the primary drawback of the power iteration is that convergence depends on the ratio $\lambda_2/\lambda_1$. This can be arbitrarily slow if this ratio is $1-\kappa$ for some $\kappa \ll 1$, i.e. $N_{PI}$ can be arbitrarily large to achieve a desired tolerance $\epsilon_a$. Moreover, the power iteration only provides lower bounds on the induced gain. Algorithm~\ref{alg:power} simply terminates if the iteration fails to make significant progress or if it reaches a maximum number of iterations. Algorithm~\ref{alg:power} provides no guarantee on the accuracy between the returned $\gamma_\pi$ and the actual value $\|G\|_{[0,T]}$ because this gap is not computed by the power iteration. The next two examples illustrate these issues. All examples are performed using MATLAB running on a desktop computer with $3$ GHz Intel core i7 processor and $16$ GB RAM. 

\begin{ex}	
	\label{ex:stateorder}
	We randomly sampled five representative SISO LTI systems for each model order $n_x = 1,10,20,\hdots, 200$. Computation times were recorded for each model to integrate a single RDE and perform one power iteration step ($G$ followed by $G^\sim$) on horizon $T=15$ sec. Each RDE integration was performed with a $\gamma>\|G\|_{[0,T]}$ so that the solution existed on $[0,T]$.  Let $T_{RDE}(n_x)$  and $T_{PI}(n_x)$ denote the computation times averaged over the five random models of state dimension $n_x$. \figref{fig:RatioOfCompTime} shows the ratio $T_{RDE}(n_x)/T_{PI}(n_x)$ along with a linear fit. As noted earlier, the costs of the RDE and power iteration scale as $\mathcal{O}(n_x^2)$ and $\mathcal{O}(n_x)$, respectively. Thus, we expect the ratio to grow as $\beta_0 + \beta_1 n_x$ for some constant $\beta_0$ and $\beta_1$. The offset and slope of the linear fit are $\beta_0 = 7.26$ and $\beta_1=0.07$. This indicates that the computational time of a single RDE integration is roughly equal to $\beta_0 + \beta_1 n_x$ power iterations.
\end{ex}

\begin{ex}
	\label{ex:powerissues}
	Let $G_1$ be an LTI system given by state matrices:
	\begin{align*}
	A=\bmtx -0.1 & 0.4 \\ -0.5 & 0 \emtx, \,\, B=\bmtx 2 \\ 0 \emtx, \,\, C_I= \bmtx 0 & 1\emtx,\,\, D_I=0
	\end{align*}
	The power iteration for induced $\mathcal{L}_2$ gain of $G_1$ (for $T=10$ sec) using $\epsilon_a = 5 \times 10^{-3}$ converged to the gain of $7.159$ in just four iterations. Next define the MIMO system $G_2$ such that $G_2 := \bsmtx G_1 & 0 \\ 0 & 0.95 G_1\esmtx$. We have $\|G_2\|_{[0,T]}=\|G_1\|_{[0,T]}$ by construction. However, the power iteration for $G_2$ converged to $7.130$ in $14$ iterations. The power iteration on $G_1$ and $G_2$ took $0.07$ sec and $0.27$ sec. \figref{fig:powerprog} shows the progress of the power iteration on these two systems.  The iteration has slow convergence for $G_2$ because the eigenvalues of $G_2^\sim G_2$ are not well separated. For comparison,  the RDE bisection returned the same bounds $[7.157,7.161]$ for both $G_1$ and $G_2$.  However, this took $11$ bisections and $3.7$ seconds of computation time.	
\end{ex}

\begin{figure}[h]
	\centering
	\includegraphics[width=0.68\linewidth]{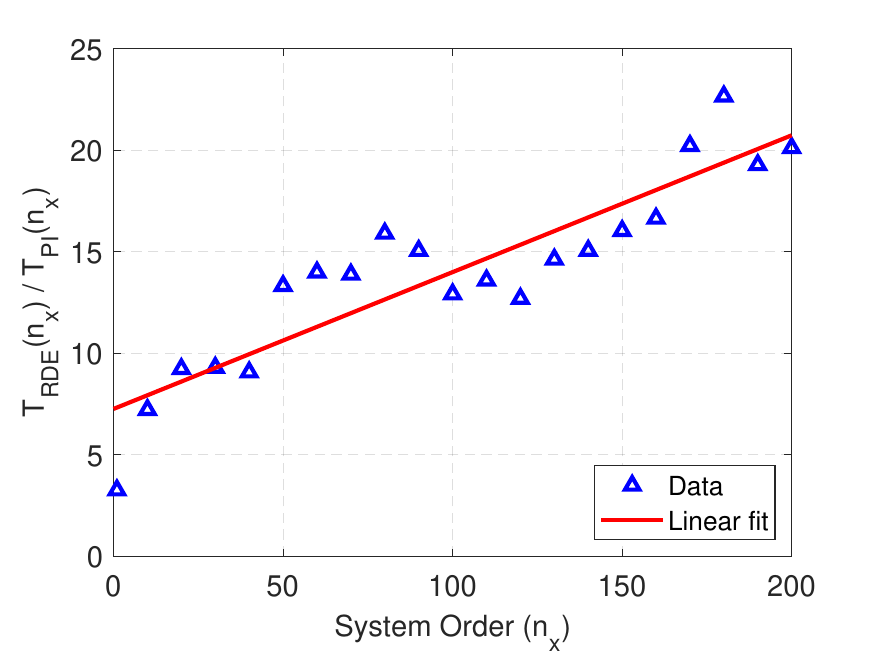}
	\caption{Ratio of Avg. Comp. Time vs System Order}
	\label{fig:RatioOfCompTime}
\end{figure}
\begin{figure}[h]
	\centering
	\includegraphics[width=0.68\linewidth]{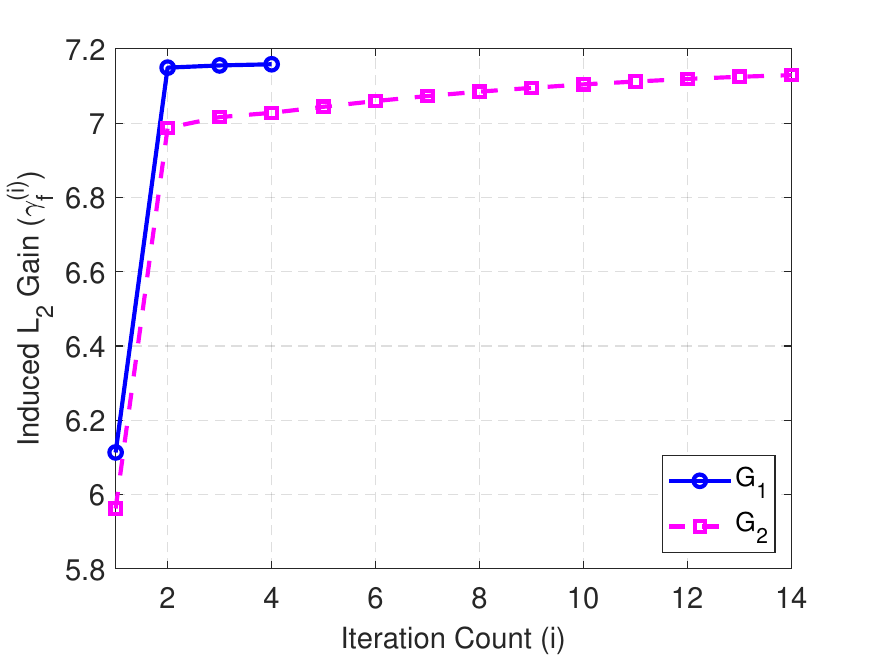}
	\caption{Power Iteration Progress}
	\label{fig:powerprog}
\end{figure}

\subsection{Combined Algorithm}
\label{sec:combAlgo}

This section presents an algorithm that combines the benefits of the power iteration and RDE bisection. The specific numerical steps are outlined in Algorithm~\ref{alg:combniedAlgo}. This algorithm is initialized with a candidate disturbance $d^{(1)}$, which is chosen randomly and normalized to have size $1$. The maximum number of power iteration $N$ is set to some high value such as $50$. The initial upper bound $\gamma_{ub} = \infty$ and lower bound $\gamma_{lb} = 0$ are fixed. The absolute tolerance $\epsilon_a$ is specified. In addition, relative tolerance $\epsilon_r$ can also be used.
\begin{algorithm}[h]
	\linespread{1}\selectfont
	\caption{Combined Algorithm} \label{alg:combniedAlgo}  
	\begin{algorithmic}[1]
		\State \textbf{Given:} $G$, $G^\sim$
		\State \textbf{Initialize:} $d^{(1)}$ with $\|d^{(1)}\|_{2,[0,T]} = 1$, $\epsilon_a$, $N$, $\gamma_{lb}$, $\gamma_{ub}$, iteration count $i=1$.
		
		\While{$\gamma_{ub} - \gamma_{lb} > \epsilon_a$}	
		
		\State \parbox[t]{0.9\linewidth}{\textbf{Run Power Iterations:}	$[\gamma_\pi$, $d_\pi]$ = Algorithm~\ref{alg:power} ($G$, $G^\sim$, $d^{(i)}$, $N$, $\epsilon_a/5$)}
		\vspace{0.02in}
		
		\State \parbox[t]{0.9\linewidth}{\textbf{Integrate RDE:} Solve RDE for $\gamma_{try} =  \gamma_\pi+\epsilon_a$.}
		\vspace{0.02in}
		
		\State \parbox[t]{0.9\linewidth}{\textbf{Update:} If $P(0)<\infty$ then $\gamma_{ub} = \gamma_{try}$, $\gamma_{lb} = \gamma_\psi$, $d_{lb} = d_\pi$ else $\gamma_{lb} = \gamma_{try}$, and construct a unit norm disturbance $d_{lb}$ using method in~\cite{iannelli2019worst} and set $d^{(i+1)} =  d_{lb}$. Increment the count $i = i + 1$.}		
		\vspace{0.02in}	
		
		\EndWhile
		
		\State \textbf{Output:} $\gamma_{lb}$, $\gamma_{ub}$, $d_{lb}$
	\end{algorithmic}
\end{algorithm}

The first step in the iterative loop is to perform power iterations with tight numerical tolerance i.e. $\epsilon_a/5$. This gives the gain 
$\gamma_\pi$ and corresponding disturbance $d_\pi$. The next step is to perform RDE integration for $\gamma_{try} = \gamma_\pi+\epsilon_a$ from respective boundary condition. Notice that this $\gamma_{try}$ is just high enough to satisfy the stopping criteria for the loop. Moreover, the tight tolerance used in power iteration increases the likelihood for existence of the RDE solution for $\gamma_{try}$. If RDE solution exists then $\gamma_{try}$, $\gamma_\pi$, $d_\pi$ are set to $\gamma_{ub}$, $\gamma_{lb}$, $d_{lb}$ respectively and we are done as no further iterations required. If RDE solution is incomplete then $\gamma_{try}$ is a better lower bound $\gamma_{lb}$ and disturbance $d_{lb}$ can be constructed from non-convergent RDE solution~\cite{iannelli2019worst}. This disturbance is set to $d^{(i+1)}$ and is used as a starting disturbance in the next iteration. Finally, the iteration count $i$ is incremented. Algorithm~\ref{alg:combniedAlgo} can be modified to switch over to the RDE bisection method, if it does not make a significant progress within some small number of RDE integration calls. This ensures computational complexity that is no worse than Algorithm~\ref{alg:bisectAlgo}. In practice, the combined algorithm: (i) is faster than RDE bisection method, and (ii) terminates with guaranteed upper and lower bounds on the induced gain within a specified numerical tolerance. Thus, it merges the benefits of RDE bisection and power iteration. 

\subsection{Special Case: $\mathcal{L}_2$-to-Euclidean Gain}
\label{sec:l2toE}

If $n_I= 0$ then the induced norm defined in Eq.~\eqref{eq:def} is a finite horizon induced $\mathcal{L}_2$-to-Euclidean gain of $G$ which is denoted by $\|G\|_{E,[0,T]}$. Note that $G:\mathcal{L}_2^{n_d}[0,T]\rightarrow\R^{n_E}$ is a finite rank operator. The next theorem presents a simpler condition to compute $\|G\|_{E,[0,\tau]}$ for any intermediate horizon $\tau\in[0,T]$. This simpler condition only requires a forward integration of the following matrix Lyapunov Differential Equation (LDE):
\begin{align}
\label{eq:LDE1}
\dot{X}(t) =  A(t) X(t) + X(t) A(t)^\top + B(t)B(t)^\top,\hspace{0.1in} X(0) = 0
\end{align}
where $X(t)\in \Sm^{n_x}$ is a state controllability Gramian of $G$ at time $t$. In comparison, the Algorithm~\ref{alg:power}, \ref{alg:bisectAlgo} or~\ref{alg:combniedAlgo} have to be run for each $\tau\in[0,T]$ to compute $\|G\|_{E,[0,\tau]}$.

\begin{theorem}	
	\label{thm:L2toE}
	Let $Y(t):=C_E(t)X(t)C_E(t)^\top$, $\forall t\in[0,T]$ be the output controllability Gramian with largest eigenvalue denoted as $\lambda_1(Y(t))$. The finite horizon induced $\mathcal{L}_2$-to-Euclidean gain of $G$ for any horizon $\tau \in [0,T]$ is given by $\| G\|_{E,[0,\tau]} = \sqrt{\lambda_1(Y(\tau))}$. Moreover, a unit-norm, worst-case disturbance $d_{wc}(t)$ for $t\in[0,\tau]$ is given by:
	\begin{align}
	\label{eq:wcdist}
	d_{wc}(t) = B(t)^\top \Phi(\tau,t)^\top C_E(\tau)^\top \frac{v_1}{\sqrt{\lambda_1(Y(\tau))}}
	\end{align}
	where $v_1$ is a unit-norm eigenvector associated with $\lambda_1(Y(\tau))$. 
\end{theorem}
\begin{proof}
	The proof is given in Appendix~\ref{sec:proof3}.  
\end{proof}
Theorem~\ref{thm:L2toE} shows that related worst-case disturbance $d_{wc}(t)$ for $t\in[0,\tau]$ can also be computed by simulating the adjoint dynamics in the direction of eigenvector $v_1$. The similar result related to the square root of the maximum eigenvalue of the output controllability Gramian appeared in~\cite{wilson1985hankel,lu1998comparison} for induced $\mathcal{L}_2$-to-$\mathcal{L}_\infty$ gain and in~\cite{corless1989improved, wilson1989convolution, rotea1993generalized} for generalized $H_2$ norm of LTI systems on infinite horizons.

%The LDE~\eqref{eq:LDE1} is equivalent to the forward integrated RDE in Theorem $1.2$ of~\cite{khargonekar1991h_} for this special case.

\section{Numerical Examples}
\label{sec:ex}
%The proposed algorithm and examples presented in this paper are available online in the LTVTools toolbox~\cite{LTVTools}.

\subsection{LTI System}
Consider SISO LTI systems of increasing system order as discussed in Example~\ref{ex:stateorder}. 
Each system is normalized so that infinite horizon $H_\infty$ norm is one. Algorithms~\ref{alg:power}, \ref{alg:bisectAlgo} and \ref{alg:combniedAlgo} are run to compute the induced $\mathcal{L}_2$ gain within $\epsilon_a = 0.01$ on a horizon $T=15$ sec. The power iteration only computes lower bounds. \figref{fig:compareSISOcompstudy} shows the computation times (averaged over $5$ random models) versus state dimension $n_x$.  The solid line denote the linear fit to the respective data. The proposed method is faster than RDE bisection and provides guaranteed upper/lower bounds (in comparison to the power iteration).
\begin{figure}[h] 
	\centering
	\includegraphics[width=0.70\linewidth]{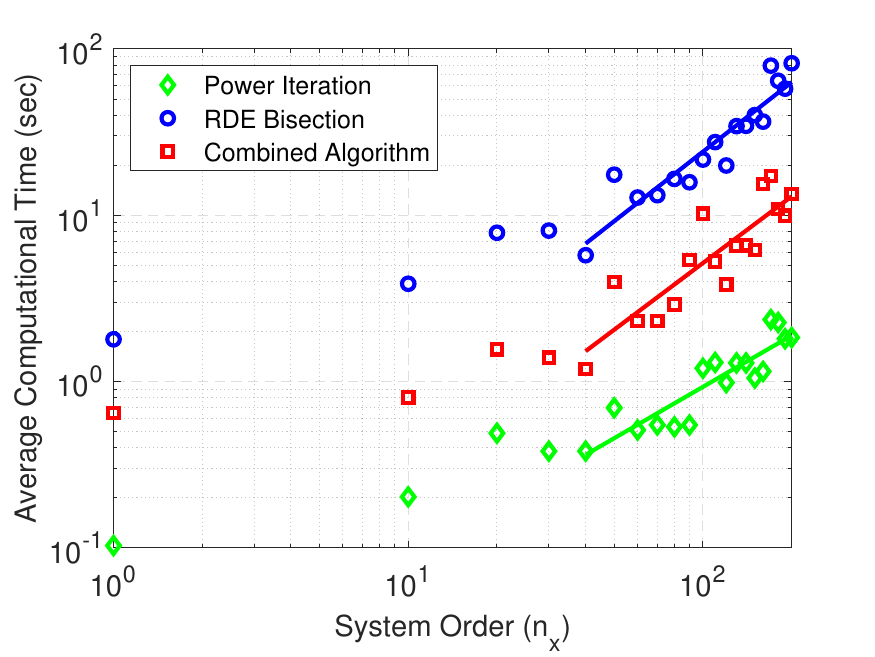}
	\caption{Average Comp. Time vs System Order}
	\label{fig:compareSISOcompstudy}
\end{figure}
\figref{fig:compareGains} shows the comparison of the bound computed using Algorithm~\ref{alg:power} and \ref{alg:bisectAlgo} for different models (sorted in increasing order of gain). It indicates that the RDE bisection method lower bound (blue circles) is always within the desired numerical tolerance $\epsilon_a$ (black line) from the upper bound (red squares). However, the power iteration lower bound (green diamonds) may or may not be within the bisection tolerance.  

\begin{figure}[h] 
	\centering
	\includegraphics[width=0.70\linewidth]{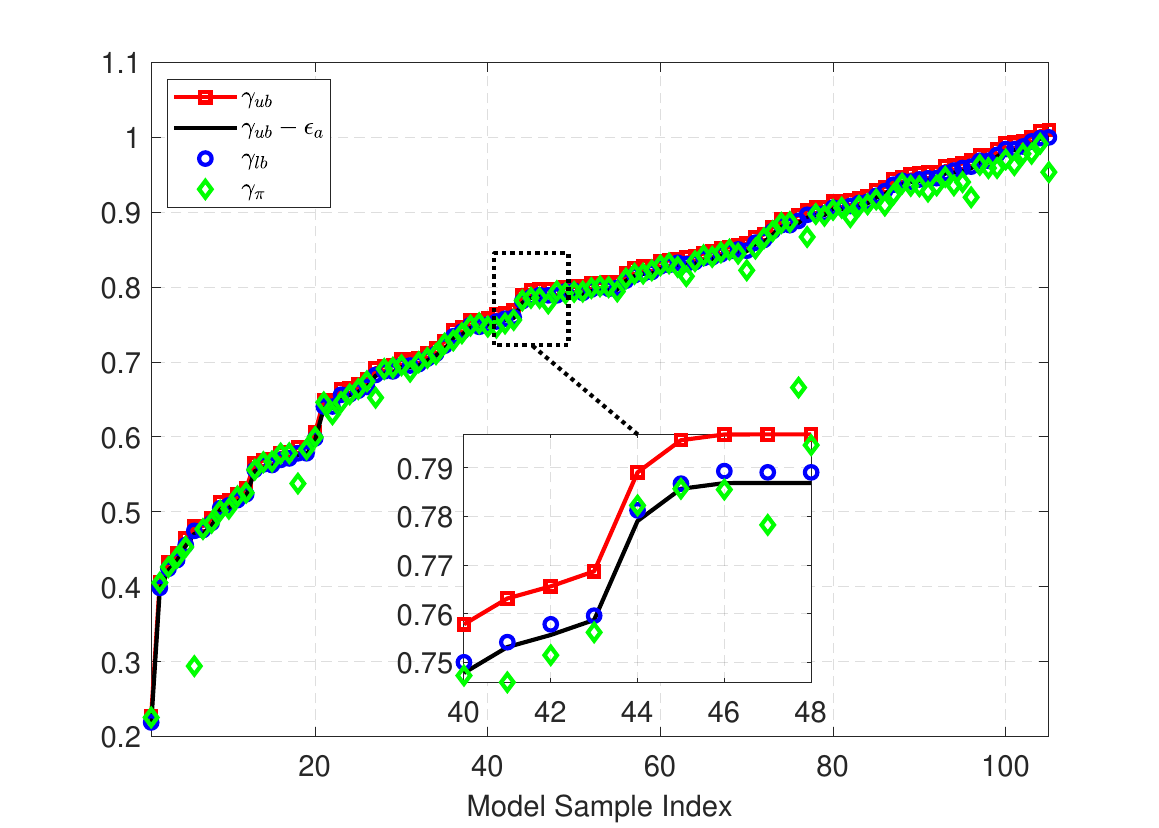}
	\caption{Comparison of Gain Computed Using Algorithm 1 and 2}
	\label{fig:compareGains}
\end{figure}

\subsection{LTV System}
Consider the Example $4.1$ in~\cite{imae1996h}. Let the LTV system be given by the following state-space matrices:
\begin{align*}
A(t) = \bmtx -1+sin(t) & 1\\0 & -4\emtx,\,\, B = C_I = \bmtx 1 & 0\\0 & 1\emtx,\,\, D_I = \bmtx 0 & 0\\0 & 0\emtx 
\end{align*}
The induced $\mathcal{L}_2$ gains are computed on the horizon of $T=10$ seconds using Algorithm~\ref{alg:power}, \ref{alg:bisectAlgo} and \ref{alg:combniedAlgo} with $\epsilon_a = 0.01$. Algorithm~\ref{alg:power} took $7$ power iterations to reach the $\gamma_\pi$ of $1.782$ in total $0.45$ seconds. Algorithm~\ref{alg:bisectAlgo} took $10$ bisections to compute the bounds of $[1.799, 1.809]$ in $4.75$ seconds. Algorithm~\ref{alg:combniedAlgo} took only two iterations to compute the same bound $[1.799, 1.809]$ in $1.63$ seconds. The unit-norm worst-case disturbances are shown in~\figref{fig:wcdist} which matches with Figure~$1$ in~\cite{imae1996h}.
\begin{figure}[h]
	\centering
	\includegraphics[width=0.70\linewidth]{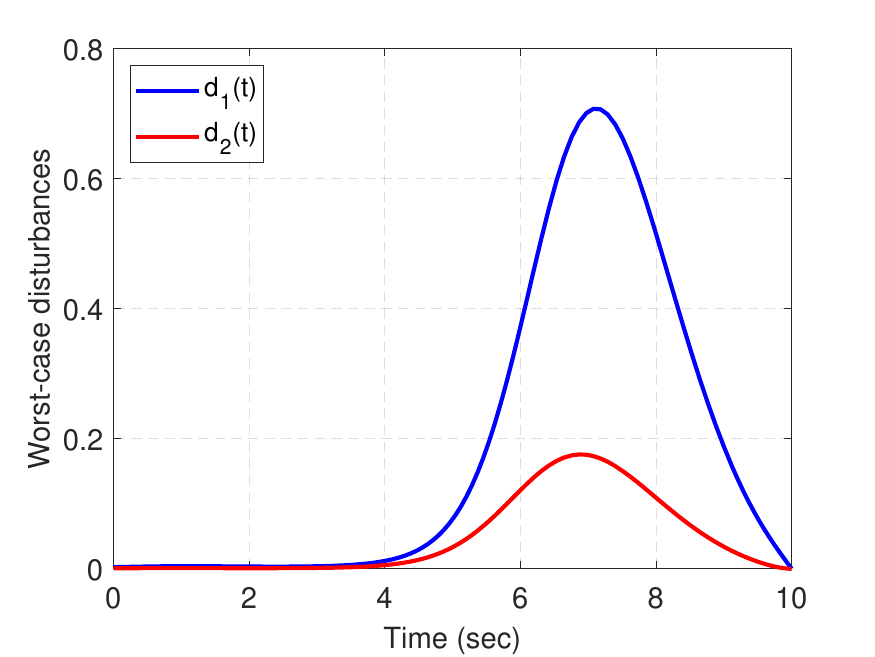}
	\caption{Worst-case Disturbances}
	\label{fig:wcdist}
\end{figure}

\subsection{Nonlinear System: Two-Link Robot}
Consider an example of a two-link robot arm as shown in the~\figref{fig:twoLinkRobot}. The mass and moment of inertia of the $i^{th}$ link are denoted by $m_i$ and $I_i$. The robot properties are $m_1=3kg$,
$m_2 = 2kg$, $l_1 = l_2= 0.3m$, $r_1=r_2 = 0.15m$,
$I_1= 0.09 kg\cdot m^2$, and $I_2= 0.06 kg\cdot m^2$. The nonlinear equations of motion \cite{murray1994mathematical} for the robot are given by:
\begin{align}
\begin{split}
\bmat{\alpha+ 2\beta\cos(\theta_2) & \delta + \beta \cos(\theta_2) \\
	\delta +  \beta \cos(\theta_2) & \delta}
\bmat{\ddot{\theta}_1 \\ \ddot{\theta}_2} +\bmat{-\beta \sin(\theta_2) \dot{\theta}_2 &
	-\beta \sin(\theta_2) (\dot{\theta}_1 + \dot{\theta}_2) \\
	\beta \sin(\theta_2) \dot{\theta}_1 & 0} \bmat{\dot{\theta}_1 \\
	\dot{\theta}_2} =\bmat{\tau_1 \\ \tau_2}  \label{eq:linkDyns}
\end{split}
\end{align}
\begin{align*}
\mbox{with } &\alpha := I_{1} + I_{2} + m_1r_1^2 + m_2(l_1^2 + r_2^2) = 0.4425    \, kg \cdot m^2\\
& \beta := m_2 l_1 r_2 = 0.09 \, kg \cdot m^2 \\
& \delta := I_{2} + m_2 r_2^2 = 0.105 \, kg \cdot m^2.
\end{align*}
The state and input are
$\eta :=[\theta_1  \ \theta_2  \ \dot{\theta}_1 \
\dot{\theta}_2]^\top$ and $\tau :=[\tau_1 \ \tau_2]^\top$,
where $\tau_i$ is the torque applied to the base of link $i$. A trajectory $\bar{\eta}$ of duration $5$ second was selected for the tip of the arm to follow. This trajectory is shown as a solid black line in \figref{fig:NomTraj}. 

\begin{figure}[h] 
	\centering
	\includegraphics[scale=0.35]{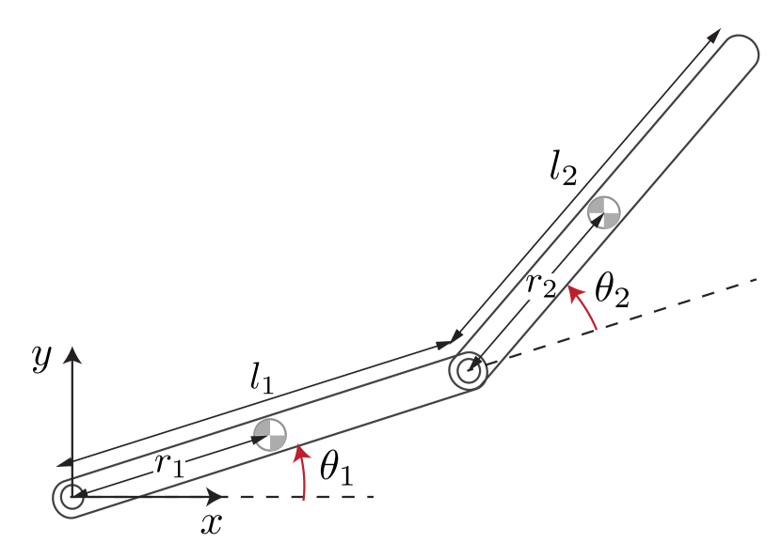}
	\caption{Two-link Planar Robot~\cite{murray1994mathematical}}
	\label{fig:twoLinkRobot}
\end{figure}

The equilibrium input torque $\bar{\tau}$ can be computed using inverse kinematics. The robot should track this trajectory in the presence of small torque disturbances $d$. The input torque vector is
$\tau = \bar{\tau} + u + d$ where $u$ is an additional control torque to reject the disturbances. The nonlinear dynamics~\eqref{eq:linkDyns} are linearized around the trajectory $(\bar{\eta}, \bar{\tau})$ to obtain an LTV system $G$:
\begin{align}
\dot{x}(t) = A(t) x(t) + B(t) \left( u(t) + d(t) \right)
\end{align}
where $x(t):=\eta(t)-\bar{\eta}(t)$ is the deviation from the equilibrium trajectory. This is an example of  ``gridded" LTV system which often arise from a linearization of a nonlinear dynamic model along a trajectory. Linear interpolation is used to approximate the system dynamics at any time during the integration. A finite horizon time-varying LQR controller is designed to reject input disturbances. The details for the control design can be found in~\cite{seiler2019finite,moore15}. We are interested in assessing the performance of the closed loop LTV system.

\begin{figure}[h]
	\centering
	\begin{minipage}{.5\textwidth}
		\centering
		\includegraphics[width=\linewidth]{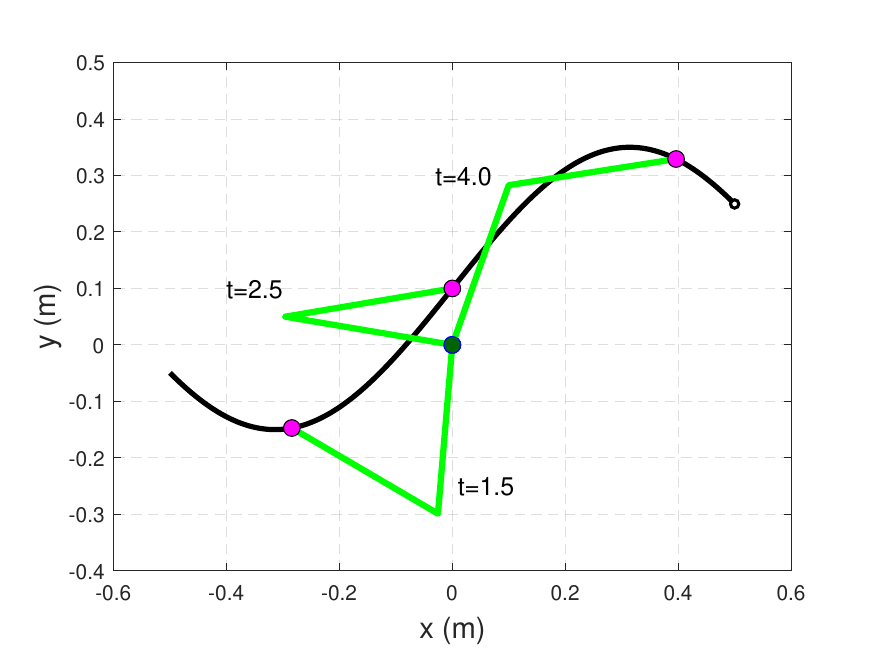}
		\caption{Nominal Trajectory with Snapshot Positions}
		\label{fig:NomTraj}
	\end{minipage}%
	\begin{minipage}{.5\textwidth}
		\centering
		\includegraphics[width=\linewidth]{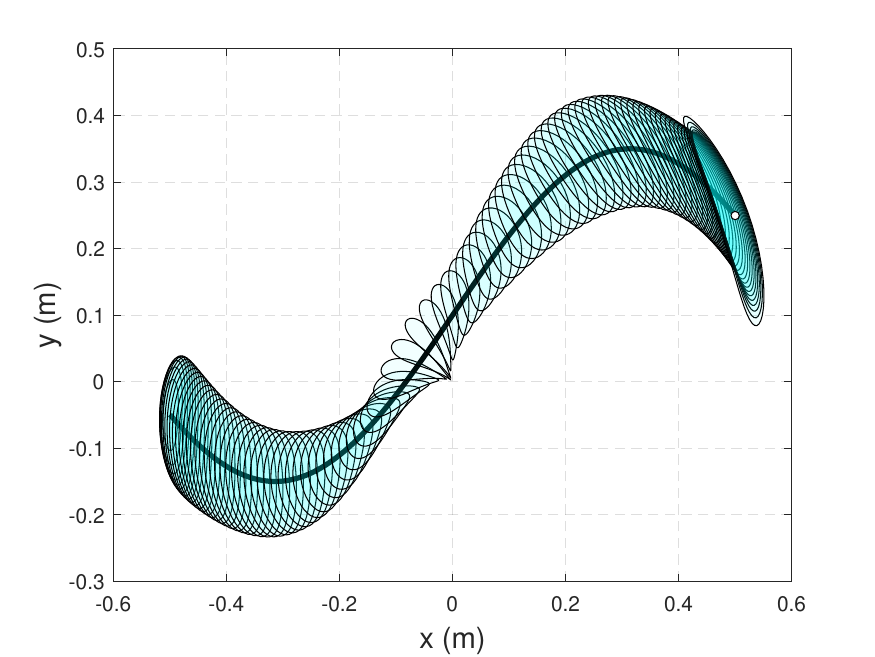}
		\caption{Euclidean Norm Bound Along the Trajectory}
		\label{fig:EuNormBound}
	\end{minipage}
\end{figure}

First, the induced $\mathcal{L}_2$ gain from disturbance input $d$ to all the linearized states $e_I = x$ is computed using Algorithm~\ref{alg:bisectAlgo} and \ref{alg:combniedAlgo} for comparison. The absolute tolerance of $\epsilon_a = 5 \times 10^{-3}$ was selected. Algorithm~\ref{alg:bisectAlgo} took $11.8$ seconds to compute the lower and upper bounds of $0.098$ and $0.102$. It took $8$ RDE bisections to achieve this accuracy. Algorithm~\ref{alg:combniedAlgo} returned the lower and upper bound of $0.096$ and $0.101$ in $5.4$ seconds. It took $1$ full iteration of proposed algorithm (i.e. $6$ power iterations and $1$ RDE integration). Next, consider $\mathcal{L}_2$-to-Euclidean gain analysis using the condition presented in Theorem~\ref{thm:L2toE}. This analysis is performed from disturbance input $d$ to the performance output $e_E = \bsmtx x_1\\ x_2 \esmtx$. It took $2.3$ seconds to compute the Euclidean norm bound along the entire trajectory for disturbance with $\|d\|_{2,[0,T]} = 5$. Using the trim values of $\theta_1$, $\theta_2$ and given robot parameters, the robot tip position in the Cartesian coordinate was computed. \figref{fig:EuNormBound} shows the cyan bound at $100$ equality spaced points around a nominal trajectory. 

Such preliminary analysis performed on a LTV system provides primary guarantees on performance. At any given time, the worst-case disturbance can also be computed using adjoint simulation. The impact of which can further be studied in the high fidelity nonlinear simulation. Moreover, this disturbance can be used to initialize the power iterations for the nonlinear model as in~\cite{tierno1997numerically}.  

\section{Conclusion}
\label{sec:conclusion}
A fast algorithm that uses complementary benefits of power iteration and RDE condition was proposed for computing finite horizon system norms. The effectiveness of this approach was demonstrated using numerical examples and computational studies. 

\appendices
\section{Compactness of $G^\sim G$}
\label{sec:compactness}
\begin{lem}
	\label{lem:compactness}
	Assume that $G$ has no feedthrough i.e. $D_I=0$, then $G$, $G^\sim$, $G^\sim G$ and $G G^\sim$ are compact operators.	
\end{lem}
\begin{proof}
	Let $\Phi(t,\tau)$ denote the unique state transition matrix of the unforced system $\dot{x}(t) = A(t) \, x(t)$ satisfying the following differential equation:
	\begin{align}
	\frac{d\Phi (t,\tau)}{dt} = A(t) \,\Phi(t,\tau), \hspace{0.2in} \Phi(\tau,\tau) = I_{n_x} 
	\end{align}
	If $D_I=0$ then $G$ can be represented by a matrix valued kernel function as 	 
	$e(t) = \int_0^t H(t,\tau) \,d(\tau) d\tau$, where $H(t,\tau) := C(t) \Phi(t,\tau) B(\tau)$ is an impulse response matrix. This impulse response matrix satisfies:
	\begin{align}
	\int_0^T \int_0^T trace\left[H(t,\tau) H(t,\tau)^\top\right] \, dt\, d\tau < \infty
	\end{align}	
	This bound follows from Theorem $3.3.1$ in~\cite{green2012linear} and the connection to the finite-horizon $H_2$ norm of $G$. It follows that $G$ is a Hilbert-Schmidt operator and hence compact (Theorem $8.8$ of~\cite{young1988introduction}).  Finally, compactness of $G$ implies that $G^\sim$, $G^\sim G$, and $GG^\sim$ are all compact (Theorems $4.8$ \& $4.10$ in~\cite{kato2013perturbation}).
\end{proof}

\section{Proof of Theorem~2}
\label{sec:proof2}

\noindent\textbf{Theorem 2.} \textit{Assume $D_I=0$ so that $G^\sim G$ is compact. Further assume that the dominant eigenvalue has multiplicity $m$ and $d^{(1)}$ satisfies $\langle d^{(1)}, \psi_k \rangle \ne 0$ for some $k\in \{1,\ldots,m\}$. Then $\gamma^{(i+1)} \ge \gamma^{(i)}$, $\forall i \geq 1$ and as $i \rightarrow \infty$ we have:
\begin{enumerate}[label=(\alph*)]
	\item $d^{(i)} \rightarrow$ span of $\{\psi_1,\ldots,\psi_m\}$
	\item $\gamma^{(i)} \rightarrow \|G\|_{[0,T]}^2$
\end{enumerate}	}
\begin{proof}	
	By Eq.~\eqref{eq:infsum}, the first iteration $r^{(1)} = G^\sim G(d^{(1)})$ yields $r^{(1)} = \sum_{k = 1}^{\infty} \lambda_k \alpha_k \psi_k$ where $\alpha_k := \langle d^{(1)}, \psi_k \rangle$. The gain $\gamma^{(1)}$ can be obtained as:
	\begin{align}
	\gamma^{(1)} &= \|r^{(1)}\|_{2,[0,T]}\nonumber\\ 
	&= \sqrt{\langle r^{(1)}, r^{(1)} \rangle}\nonumber\\
	&= \sqrt{\bigg\langle \sum_{k = 1}^{\infty} \lambda_k \alpha_k \psi_k, \sum_{j = 1}^{\infty} \lambda_j \alpha_j \psi_j\bigg\rangle}\nonumber\\
	&= \sqrt{\sum_{k = 1}^{\infty} \sum_{j = 1}^{\infty} \lambda_k \lambda_j \alpha_k \alpha_j  \langle \psi_k,  \psi_j\rangle}
	\end{align}
	Note $\langle \psi_k, \psi_j \rangle = 0$ for $k\neq j$ and $\langle \psi_k, \psi_j \rangle = 1$ for $k=j$, i.e.
	\begin{align}
	\gamma^{(1)} = \sqrt{\sum_{k = 1}^{\infty} (\lambda_k \alpha_k)^2}
	\end{align}
	The next disturbance $d^{(2)}$ can be obtained by:
	\begin{align}
	d^{(2)} = \frac{r^{(1)}}{\gamma^{(1)}} = \frac{1}{\gamma^{(1)}} \sum_{k = 1}^{\infty} \lambda_k \alpha_k \psi_k
	\end{align}
	Another application of $G^\sim G$ yields $r^{(2)} = G^\sim G(d^{(2)})$, i.e.
	\begin{align}			
	r^{(2)} &= \sum_{k = 1}^{\infty} \lambda_k \langle d^{(2)}, \psi_k \rangle \psi_k\nonumber\\
	&= \frac{1}{\gamma^{(1)}} \sum_{k = 1}^{\infty} \lambda_k \bigg\langle \sum_{j = 1}^{\infty} \lambda_j \alpha_j \psi_j, \psi_k \bigg\rangle \psi_k\nonumber\\
	&= \frac{1}{\gamma^{(1)}} \sum_{k = 1}^{\infty} \sum_{j = 1}^{\infty}  \lambda_k \lambda_j \alpha_j\langle \psi_j, \psi_k \rangle \psi_k\nonumber\\
	&= \frac{1}{\gamma^{(1)}} \sum_{k = 1}^{\infty} \lambda_k^2 \alpha_k \psi_k
	\end{align}
	Define the constant $c^{(i)} := 1/(\gamma^{(i-1)} \hdots \gamma^{(2)} \gamma^{(1)})$ and repeat above steps to obtain the following expressions for the $i^{th}$ iteration as:
	%Define the constant $c^{(i)} := 1/(\gamma^{(i-1)} \hdots \gamma^{(2)} \gamma^{(1)})$ and use Eq.~\eqref{eq:infsum} to obtain the following expressions for $i^{th}$ iteration:
	\begin{align}
	\label{eq:rgamma}
	r^{(i)} = c^{(i)} \, \sum_{k = 1}^{\infty} \lambda_k^i \alpha_k \psi_k, \hspace{0.3in}	
	\gamma^{(i)} = c^{(i)}\, \sqrt{\sum_{k=1}^\infty (\lambda_k^i \alpha_k)^2}
	\end{align}
	The unit-norm disturbance at the $i^{th}$ iteration is given by:
	\begin{align}
	\label{eq:dith}
	d^{(i)} = \frac{r^{(i-1)}}{\gamma^{(i-1)}} = c^{(i)} \, \sum_{k = 1}^{\infty} \lambda_k^{i-1} \alpha_k \psi_k
	\end{align}
	Since the dominant eigenvalue $\lambda_1$ has multiplicity $m$ we have:
	\begin{align}
	\label{eq:dist}
	d^{(i)} = c^{(i)}\, \lambda_1^{i-1}\left[ \sum_{k = 1}^{m} \alpha_k \psi_k + \sum_{k = m+1}^{\infty} \bigg(\frac{\lambda_k}{\lambda_1}\bigg)^{i-1} \alpha_k \psi_k \right]
	\end{align}
	As $i\rightarrow \infty$ the second term in the above sum $\rightarrow 0$ due to $\frac{\lambda_k}{\lambda_1} < 1$. Moreover, the assumption $\alpha_k \neq 0$ for some $k = 1, 2, \hdots, m$ ensures that $d^{(i)}$ converges to the subspace spanned by dominant eigenvectors. This assumption holds with probability 1 if $d^{(1)}$ is chosen randomly. Note that the convergence rate depends on ratio $\frac{\lambda_{m+1}}{\lambda_1}$.	
	
	Next, use the expression for $\gamma^{(i)}$ to write:
	\begin{align}
	\label{eq:gammai}
	\gamma^{(i)} = c^{(i)}\,\lambda_1^i \, \sqrt{\sum_{k=1}^m \alpha_k^2 + \sum_{k=m+1}^\infty \bigg(\frac{\lambda_k}{\lambda_1}\bigg)^{2i} \alpha_k^2}
	\end{align}
	As $i\rightarrow \infty$, the second term inside the square root $\rightarrow 0$. Moreover, the disturbance is normalized at each iterate, i.e. $\|d^{(i)}\|_{2,[0,T]}=1$. It follows from Eq.~\eqref{eq:dist} that: 
	\begin{align}
	\label{eq:eq29}
	c^{(i)}\, \lambda_1^{i-1} \sqrt{\sum_{k = 1}^{m} \alpha_k^2} \rightarrow 1 \mbox{ as } i \rightarrow \infty
	\end{align}
	Combine \eqref{eq:eq29} and \eqref{eq:gammai} to conclude $\gamma^{(i)}\rightarrow \lambda_1$. Theorem $6.10$ of~\cite{hunter2001applied} implies that induced norm of $G$ is equal to $\sqrt{\lambda_1(G^\sim G)}$. Thus, $\gamma^{(i)} \rightarrow \|G\|_{[0,T]}^2$ as $i \rightarrow \infty$.
	
	Finally, use Eq.~\eqref{eq:rgamma} to write the ratio of gain for two subsequent iterations as:
	\begin{align}
	\label{eq:relation}
	\frac{\gamma^{(i+1)}}{\gamma^{(i)}} = \frac{c^{(i+1)}\sqrt{\sum_{k=1}^\infty (\lambda_k^{i+1} \alpha_k)^2}}{c^{(i)} \sqrt{\sum_{k=1}^\infty (\lambda_k^{i} \alpha_k)^2}} = \frac{\sqrt{\sum_{k=1}^\infty (\lambda_k^{i+1} \alpha_k)^2}}{c^{(i)} \sum_{k=1}^\infty (\lambda_k^{i} \alpha_k)^2}
	\end{align}	
	Define $f_k := \lambda_k^{i+1} \alpha_k$ and $g_k:=\lambda_k^{i-1} \alpha_k$ to obtain:
	\begin{align}
	\label{eq:holder}
	\begin{split}
	c^{(i)} \sum_{k=1}^\infty (\lambda_k^{i} \alpha_k)^2 &= c^{(i)} \sum_{k=1}^\infty f_k\, g_k \leq c^{(i)} \|f\|_2 \|g\|_2 \hspace{0.1in} (\mbox{H\"{o}lder's Inequality})
	\end{split}
	\end{align}
	It follows from Eq.~\eqref{eq:dith} that $\|d^{(i)}\|_{2,[0,T]} = c^{(i)} \|g\|_2$ and hence $c^{(i)} \|g\|_2=1$ as the disturbances have unit norm.	It follows from \eqref{eq:relation} and \eqref{eq:holder} that $\gamma^{(i+1)} \ge \gamma^{(i)}$, $\forall i \geq 1$. 
\end{proof}
\begin{eremark}
	For the LTI systems, compactness of $G$ (and hence $G^\sim G$) can be used to arrive a finite-dimensional eigenvalue problem, as in the periodic sampled-data systems frequency-response literature~\cite{yamamoto1996frequency}, \cite{ito2001bisection},\cite{dullerud1999computing}.
\end{eremark}

\section{Proof of Theorem~3}
\label{sec:proof3}

\noindent\textbf{Theorem 3.} \textit{Let $Y(t):=C_E(t)X(t)C_E(t)^\top$, $\forall t\in[0,T]$ be the output controllability Gramian with largest eigenvalue denoted as $\lambda_1(Y(t))$. The finite horizon induced $\mathcal{L}_2$-to-Euclidean gain of $G$ for any horizon $\tau \in [0,T]$ is given by $\| G\|_{E,[0,\tau]} = \sqrt{\lambda_1(Y(\tau))}$. Moreover, a unit-norm, worst-case disturbance $d_{wc}(t)$ for $t\in[0,\tau]$ is given by:
\begin{align*}
%\label{eq:wcdist}
d_{wc}(t) = B(t)^\top \Phi(\tau,t)^\top C_E(\tau)^\top \frac{v_1}{\sqrt{\lambda_1(Y(\tau))}}
\end{align*}
where $v_1$ is a unit-norm eigenvector associated with $\lambda_1(Y(\tau))$.}
\begin{proof}	
	The proof is similar to existing results on Gramian-based minimum
	energy control (Theorem~$1$, Section~$22$ of~\cite{brockett2015finite}).
	The results in~\cite{brockett2015finite} provide a condition for transferring the state from
	a given initial state $x(0)=x_0$ to a final state
	$x(T)=x_T$ using least amount of control energy. These results can be
	used to determine the input of unit norm that maximizes the Euclidean
	norm of the final state starting from zero initial conditions. The
	proof is given below for completeness.
	
	Let $G:\mathcal{L}^{n_d}_{2}[0,\tau] \rightarrow \R^{n_E}$ be a given bounded linear operator with the adjoint $G^\sim:\R^{n_E} \rightarrow \mathcal{L}^{n_d}_{2}[0,\tau]$ for any $\tau \in [0,T]$. The notation $\mathcal{R}(G)$ and $\mathcal{N}(G)$ are used to denote the range and null space of $G$ respectively. We know that $G$ is a finite rank operator because the $\mathcal{R}(G)$ is finite-dimensional. Moreover, every bounded finite rank operator is compact (Theorem $8.1-4$ of \cite{kreyszig1978introductory}). Assume $G$ is output-controllable, thus $\mathcal{R}(G) \equiv \R^{n_E}$ which is closed. This implies  $\mathcal{R}(G^\sim)$ is also closed (Chapter $4$, Theorem $5.13$ in \cite{kato2013perturbation}). In this case we have:
	\begin{align}
	\label{eq:rangenull}
	\mathcal{R}(G^\sim) = \mathcal{N}(G)^{\perp}
	\end{align}
	where $\mathcal{N}(G)^{\perp}$ denote an orthogonal complement of $\mathcal{N}(G)$. Note that $\mathcal{N}(G)$ is a closed linear subspace of the Hilbert space $\mathcal{L}^{n_d}_{2}[0,\tau]$ (Theorem $1.18$ of~\cite{rudin1973functional}). Thus, we can decompose the Hilbert space $\mathcal{L}^{n_d}_{2}[0,\tau]$ as $\mathcal{N}(G)^{\perp} \oplus \mathcal{N}(G)$ (Section $3.4$ of~\cite{luenberger1997optimization}, Section $5.1$ of~\cite{kato2013perturbation}). This implies that any $d \in \mathcal{L}^{n_d}_{2}[0,\tau]$ can be decomposed as $d = d_1 + d_2$, where $d_1 \in \mathcal{N}(G)^{\perp}$, $d_2 \in \mathcal{N}(G)$ and the inner product $\langle d_1, d_2 \rangle$ = 0. Moreover, since $\mathcal{N}(G)$ is a null space of $G$, we have $G(d_2) = 0$. Use the the linearity property of an operator i.e. $G(d_1 + d_2) = G(d_1) + G(d_2)$ to rewrite the square of the induced $\mathcal{L}_2$-to-Euclidean cost as follows:
	\begin{align}
	J(d) &:=  \frac{\| e(\tau)\|_2^2}{\|d\|^2_{2,[0,\tau]}} = \frac{e(\tau)^\top e(\tau)}{\langle d, d \rangle} = \frac{G(d)^\top G(d)}{\langle d, d \rangle} =  \frac{G(d_1)^\top G(d_1)}{\langle d_1, d_1 \rangle + \langle d_2, d_2 \rangle } \leq \frac{G(d_1)^\top G(d_1)}{\langle d_1, d_1 \rangle}
\end{align}
	This means $J(d) \leq J(d_1)$. Thus, the optimization over infinite dimensional Hilbert space $\mathcal{L}^{n_d}_{2}[0,\tau]$ is equivalent to the optimization over $ \mathcal{N}(G)^{\perp}$. Use Eq.~\eqref{eq:rangenull} to rewrite the problem as:
	\begin{align}
	\|G\|^2_{E,[0,\tau]} &= \sup_{d\in\mathcal{L}^{n_d}_{2}[0,\tau]} J(d)\nonumber\\
	&= \sup_{d \in \mathcal{N}(G)^\perp} J(d) = \sup_{d \in \mathcal{R}(G^\sim)} J(d)
	\end{align}
	Define $w:=e(\tau)$ to rewrite the optimization problem over the finite dimensional space $\R^{n_E}$ as:	
	\begin{align}
	\|G\|^2_{E,[0,\tau]} = \sup_{w \in \R^{n_E}} \frac{\langle GG^\sim(w), GG^\sim(w)\rangle}{\langle G^\sim (w), G^\sim(w) \rangle} \hspace{0.13in} \label{eq:prob1}
	\end{align}
	Using $G^\sim(w) = B(t)^\top \Phi(\tau,t)^\top C_E(\tau)^\top w$, the composition $GG^\sim(w)$ can be written as:
	\begin{align}
	GG^\sim(w) = C_E(\tau)X(\tau)C_E(\tau)^\top\, w = Y(\tau) \,w
	\end{align}		
	where $X(\tau) = \int_{0}^{\tau} \Phi(\tau,s) B(s) \, B(s)^\top \Phi(\tau,s)^\top ds$ is a solution to LDE~\eqref{eq:LDE1} at time $\tau$. The inner product in the denominator of Eq.~\eqref{eq:prob1} can be written as:
	\begin{align}
	\langle G^\sim (w), G^\sim(w) \rangle = \langle w, GG^\sim(w) \rangle = w^\top \, Y(\tau) \,w
	\end{align}
	Thus, the optimal cost from Eq.~\eqref{eq:prob1} can be written as:
	\begin{align}
	\label{eq:L2toECost}
	\|G\|^2_{E,[0,\tau]} &= \sup_{w\in\R^{n_E}} \frac{w^\top Y(\tau)^\top Y(\tau) w}{w^\top Y(\tau) w}
	\end{align}
	where $Y(\tau) \in \Sm^{n_E}$ is diagonalizable and system is assumed to be output controllable i.e. $Y(\tau) > 0$. The eigenvalue decomposition of $Y(\tau)$ can be obtained as $V \Lambda V^\top$ where $\Lambda > 0$ is a diagonal matrix of eigenvalues and $V$ is a unitary matrix whose columns are orthonormal eigenvectors $v_i$, for $i = 1,2,\hdots,n_E$. Substitute this relation in Eq.~\eqref{eq:L2toECost} and use $V^\top V = I$ to obtain:
	\begin{align}
	\|G\|^2_{E,[0,\tau]} &= \sup_{w\in\R^{n_E}} \frac{w^\top V \Lambda \Lambda V^\top w}{w^\top V \Lambda V^\top w}
	\end{align}
	Define $u:=\Lambda^{\frac{1}{2}} V^\top w$ to rewrite the optimization problem as:
	\begin{align}
	\|G\|^2_{E,[0,\tau]} &= \sup_{u\in\R^{n_E}} \frac{u^\top \Lambda u}{u^\top u}
	\end{align}
	This is an eigenvalue problem and the cost is maximized with the optimal solution $u^* = \bsmtx 1 & 0 & \hdots & 0\esmtx^\top$ where $u^*$ is an eigenvector of the diagonal matrix $\Lambda$ corresponding to the maximum eigenvalue $\lambda_1$, which is also the maximum eigenvalue of the matrix $Y(\tau)$. Thus $\|G\|^2_{E,[0,\tau]} = \lambda_1 (Y(\tau))$. Taking square root implies $\|G\|_{E,[0,\tau]} = \sqrt{\lambda_1(Y(\tau))}$. The worst-case disturbance in $\mathcal{R}(G^\sim)$ is obtained by simulating the adjoint dynamics with $w^* =  V \Lambda^{-\frac{1}{2}} u^* = \frac{v_1}{\sqrt{\lambda_1}}$, i.e. $d_{wc}(t) = B(t)^\top \Phi(\tau,t)^\top C_E(\tau)^\top \frac{v_1}{\sqrt{\lambda_1(Y(\tau))}}$.
\end{proof}

\section*{Acknowledgment}
We would like to thank our collaborators Prof. Andrew Packard and graduate students Kate Schweidel, Emmanuel Sin and Alex Devonport at the University of California, Berkeley for their valuable feedback. We also thank Dr. Douglas Philbrick from Naval Air Warfare Center Weapons Division (NAWCWD) at China Lake for helpful discussions.

%\nocite{*}
\bibliographystyle{ieeetr}
\bibliography{References}

\end{document}